\documentclass[11pt,a4paper]{amsart}
\usepackage[left=3.5cm,right=3.5cm,top=3.5cm,bottom=3.5cm]{geometry}
\usepackage{enumerate}
\usepackage{leftidx}
\usepackage{mathtools}
\usepackage{amsmath,amscd,amssymb,amsthm}
\usepackage[arrow,matrix]{xy}
\usepackage[OT2,T1]{fontenc}
\usepackage{booktabs}
\usepackage[position=bottom]{subfig}
\usepackage{lipsum}
\usepackage{graphicx}
\usepackage{hyperref}
\usepackage{setspace}

\DeclareMathOperator{\gal}{Gal}

\DeclareMathOperator{\AGL}{AGL}

\theoremstyle{definition}
\newtheorem{definition}{Definition}[section]

\newtheorem{remark}[definition]{Remark}
\newtheorem*{remark*}{Remark}

\theoremstyle{plain}
\newtheorem{theorem}[definition]{Theorem}

\newtheorem{lemma}[definition]{Lemma}
\newtheorem{proposition}[definition]{Proposition}

\newcommand{\F}{{ \mathbb F }}
\newcommand{\N}{{ \mathbb N }}

\author[A. Dukes]{Austin Dukes}
\address{University of South Florida\\
4202 E Fowler Ave\\
33620 Tampa, US.}
\email{austindukes@usf.edu}

\author[A. Ferraguti]{Andrea Ferraguti}
\address{Scuola Normale Superiore, Piazza dei Cavalieri 7, 56126 Pisa}
\email{andrea.ferraguti@sns.it}

\author[G. Micheli]{Giacomo Micheli}
\address{University of South Florida\\
4202 E Fowler Ave\\
33620 Tampa, US.
}
\email{gmicheli@usf.edu}
\title{Optimal Selection for Good polynomials of degree up to five}

\makeatletter
\@namedef{subjclassname@2020}{%
  \textup{2020} Mathematics Subject Classification}
\makeatother
\begin{document}

\begin{abstract}
Good polynomials are the fundamental objects in the Tamo-Barg constructions of Locally Recoverable Codes (LRC). In this paper we classify all good polynomials up to degree $5$, providing explicit bounds on the maximal number $\ell$ of sets of size $r+1$ where  a polynomial of degree $r+1$ is constant, up to $r=4$. This directly provides an explicit estimate (up to an error term of $O(\sqrt{q})$, with explict constant) for the maximal length and dimension of a Tamo-Barg LRC. Moreover, we explain how to construct good polynomials achieving these bounds. Finally, we provide computational examples to show how close our estimates are to the actual values of $\ell$, and we explain how to obtain the best possible good polynomials in degree $5$.
\end{abstract}

\maketitle

\section{Introduction}

A code $\mathcal{C}$ with length $n$ and dimension $k$ is called a locally recoverable code (LRC) with locality $r$, or a $(n, k, r)$-LRC, if, for any $v = (v_1, \cdots, v_n) \in \mathcal{C}$ and any $1 \leq i \leq n$, the coordinate $v_i$ is a function of at most $r$ other coordinates $v_{i_1}, v_{i_2}, \ldots, v_{i_r}$ of $v$. In other words, the value of any symbol in a particular codeword can be recovered by accessing at most $r$ other symbols of the codeword.

Given the linear $(n, k, r)$-LRC $\mathcal{C}$,
Gopalan \textit{et al.} \cite{Gop} and
Papailiopoulos and Dimakis \cite{PaDi} proved that the minimum distance $d(\mathcal{C})$ of $\mathcal{C}$ satisfies the upper bound $d(\mathcal{C}) \leq n - k - \lceil k/r \rceil + 2$. As in the literature, we will say $\mathcal{C}$ is an \textit{optimal} LRC if the minimum distance $d(\mathcal{C})$ of $\mathcal{C}$ achieves this bound, that is, if $d(\mathcal{C}) = n - k - \lceil{k/r} \rceil + 2$.

A powerful approach to constructing LRCs was given by Tamo and Barg in \cite{TaBa}, and it can be accomplished by constructing polynomials of degree $r+1$ which are constant on pairwise disjoint subsets of $\mathbb{F}_q$ of size $r+1$. Such polynomials are called \textit{good} polynomials.

More formally, for a nonnegative integer $\ell$ the polynomial $f \in \mathbb{F}_q[X]$ is said to be $(r, \ell)$-good if
\begin{itemize}
\item the degree of $f$ is $r+1$, and
\item there are pairwise disjoint sets $A_1, \ldots, A_{\ell} \subseteq \mathbb{F}_q$, each of cardinality $r+1$, such that $f(A_i) = \{t_i\}$ for some $t_i \in \mathbb{F}_q$, i.e., $f$ is constant on $A_i$.
\end{itemize}

Given a good polynomial, one can construct an optimal linear LRC as follows (we use the notation of \cite{Mes}, which is the most convenient for our purposes). Fix $r \geq 1$, and let $f(X) \in \mathbb{F}_q$ be a good polynomial. Write $n = (r+1)\ell$ and $k = rt$, where $t \leq \ell$. For $a = (a_{ij} \mid i = 0, \ldots, r-1; j = 0, \ldots t-1) \in \left(\mathbb{F}_q\right)^k$, define encoding polynomials $$f_a(X) = \sum_{i=1}^{r-1} \sum_{j=0}^{t-1} a_{ij} f(X)^j X^i.$$ Let $A = \bigcup_{i=1}^{\ell} A_i$ and define $$\mathcal{C} = \left\{(f_a(x), x \in A) \mid a \in \left(\mathbb{F}_q\right)^k\right\}.$$ Then $\mathcal{C}$ is an optimal linear $(n, k, r)$-LRC code over $\mathbb{F}_q$.

In the rest of the paper we classify all $(r,\ell)$-good polynomials up to $r=4$ as follows:
for any fixed prime power $q$ (even or odd) and a fixed $r$ up to $4$, we provide an explicit estimate (of the form $cq+O(\sqrt{q})$, where $c\in [0,1)$, and the implied constant in the error term is explicitly computable) on the maximal $\ell$ such that a polynomial of degree $r+1$ is $(r,\ell)$-good. Moreover, we provide examples of polynomials achieving these values for $\ell$, showing that the estimate is the best possible. 

The machinery we use involves Galois theory, the classification of transitive subgroups of the symmetric group $\mathcal S_n$ up to $n=5$, and the theory of function fields, using the results and techniques of \cite{micheli2019selection,micheli}, then further developed in \cite{bartoli2021construction, bartoli2020algebraic, bartoli2020r, bartoli2021investigating, permutationFerramicheli, ferraguti2021exceptional}.
\section{Monodromy groups and totally split places}
Let $q$ be a prime power and $f\in \F_q[X]$ with $f\notin \F_q[X^p]$, where $p$ is the characteristic of $\F_q$. Throughout the paper, we will call such polynomials separable, for short. Let $t$ be transcendental over $\F_q$.
\begin{definition}
The \emph{arithmetic monodromy group} of $f$, denoted by $A(f)$, is the Galois group of $f(X)-t$ seen as a polynomial over $\F_q(t)$. The \emph{geometric monodromy group} of $f$, denoted by $G(f)$, is the Galois group of $f(X)-t$ seen as a polynomial over $\overline{\F}_q(t)$. 
\end{definition}
It is easy to see for any non-constant $f$, the polynomial $f-t$ is geometrically irreducible. Hence both the arithmetic and the geometric monodromy groups are isomorphic to transitive subgroups of the symmetric group of degree $\deg f$, after choosing a labeling of the roots of $f-t$ in $\overline{\F_q(t)}$. Different labelings yield conjugate embeddings.

Recall that $G(f)$ is a normal subgroup of $A(f)$, and the quotient is isomorphic to the Galois group of the extension $(M\cap \overline{\F}_q)/\F_q$, where $M$ is the splitting field of $f-t$.

For any fixed $n$, it is possible to construct a polynomial $f$ of degree $n$ having $A(f)=G(f)=S_n$. Hence one can define a function

$$G_n(\cdot)\colon \{\mbox{prime powers}\}\to \N$$
that assigns to every prime power $q$ the least positive integer such that there exists a separable $f\in \F_q[X]$ of degree $n$ with $|G(f)|=|A(f)|=G_n(q)$. Notice that $G_n(q)\geq n$ for every $q$, because a transitive group of degree $n$ must have at least order $n$.

Thanks to the the techniques introduced in \cite{micheli}, given a separable polynomial $f\in \F_q[X]$ such that $A(f)=G(f)$, one can obtain an explicit estimate on cardinality of the set
$$\mathcal T^1_{split}(f)\coloneqq\{t_0\in\F_q\colon f(X)-t_0 \mbox{ splits into $\deg f$ distinct linear factors}\}.$$

This is done via the following result:

\begin{proposition}{{\cite[Proposition 3.1]{micheli}}}\label{split_places}
Let $f\in \F_q[X]$ be a separable polynomial of degree $n$ with $G(f)=A(f)$ and let $g_f$ be the genus of the splitting field of $f(X)-t$. Then
$$\frac{q+1-2g_f\sqrt{q}}{|G(f)|}-\frac{n}{2}\leq \mathcal |T^1_{split}(f)|\leq \frac{q+1+2g_f\sqrt{q}}{|G(f)|}$$
\end{proposition}
The genus $g_f$ can be bounded solely in terms of $\deg f$, using for example Castelnuovo's inequality. As noticed in \cite[Proposition 3.3]{micheli}, if $\text{char } \F_q\nmid |G(f)|$ we have $g_f\leq \frac{(n-2)|G(f)|+2}{2}$. 

It is clear from the above proposition that for a fixed $n$, minimizing $|G(f)|$ maximizes the expected number of totally split places, which in turn maximises the dimension of the Tamo-Barg code.

In this paper, we compute the function $G_n$ for every $n\in \{2,\ldots,5\}$. The simpler cases $n=2,3,4$ are completely treated in Section \ref{small_degrees}. When $n=5$ the problem becomes more difficult as, up to conjugation, there are 5 transitive subgroups of the symmetric group $S_5$:

\begin{itemize}
\item The cyclic group $C_5$, generated by a $5$-cycle;
\item The dihedral group $D_5$, generated by a $5$-cycle and a product of two disjoint transpositions;
\item The affine general linear group $\AGL_1(\F_5)$, isomorphic to $C_5\rtimes C_4$, generated by a $5$-cycle and a $4$-cycle.
\item The alternating group $A_5$;
\item The symmetric group $S_5$.
\end{itemize}

Nevertheless, we are able to prove the following theorem for good polynomials of degree $5$.
\begin{theorem}\label{main_thm}
Let $q$ be a prime power. Then:

$$G_5(q)=\begin{cases}5 & \mbox{if }5\mid q(q-1)\\
				10 & \mbox{if }5\mid q+1\\
				120 &\mbox{otherwise}\\\end{cases}.$$
\end{theorem}
As we mentioned above, the estimate of Proposition \ref{split_places} can be made explicit, leading to a formula for the maximal dimension of a Tamo-Barg code of locality $4$.
When $G(f)=A(f)\cong D_5,C_5$ the genus of the splitting field of $f$ is necessarily zero (by Hurwitz formula, for example), and the error term is therefore $O(1)$. In the next theorem we give an explicit estimate for the remaining cases, restricting for simplicity to the case $2,3,5\nmid q$.

\begin{theorem}
Let $q$ be a prime power with $2,3,5\nmid q$. Let $f\in \F_q[X]$ of degree $5$ with $G(f)=A(f)\cong S_5$. Then:
$$\frac{q+1-72\sqrt{q}}{120}-\frac{5}{2}\leq |\mathcal T^1_{split}(f)|\leq \frac{q+1+72\sqrt{q}}{120}.$$
\end{theorem}
\begin{proof}
Let $x$ be a root of $f(X)-t$. Let $F\coloneqq \overline{\F}_q(x)$ and $M$ be the splitting field of $f-t$ over $\overline{\F}_q(t)$. By Proposition \ref{split_places}, all we have to do is bound the genus $g_M$ of $M$. We will do it via Riemann-Hurwitz applied to the degree $24$ extension $M/F$. Notice that $F$ has genus $g_F$ equal to zero. We have that:
$$2g_M-2=24(2g_F-2)+\sum_{P}\sum_{Q|P}(e_{Q|P}-1),$$
because thanks to our assumptions on $q$ the ramification is tame. Here the external sum is over all places $P$ of $F$, while the internal one is over all places $Q$ of $M$ dividing $P$, and $e(Q|P)$ is the ramification index. Since $f$ has degree $5$, its derivative has degree $4$ and therefore there are at most $5$ places of $F$ that can ramify in $M$ (notice that a place of $\overline{\F}_q(t)$ ramifies in $M$ if and only if it ramifies in $F$). Now since $M/F$ is a Galois extension the ramification index $e(Q|P)$ depends only on $P$, and it is at most $24$. On the other hand, there are at least two places of $F$ that ramify in $M$, since there are at least two places of $\overline{\F}_q(t)$ that ramify in $F$: the infinite place and a finite one, since the derivative of $f$ has positive degree. All in all, we have that
$$\sum_P\sum_{Q|P}(e_{Q|P}-1)= \sum_Pe_{Q|P}-\sum_P\sum_{Q|P}1\leq 5\cdot 24-2=118,$$
and substituting in the above equation yields $g_M\leq 36$.
\end{proof}

\section{Degrees up to $4$}\label{small_degrees}
In this section we compute $G_2,G_3$ and $G_4$. We start with two general lemmas.
\begin{lemma}\label{additive_case}
Let $p$ be a prime and $q=p^m$ for some $m\geq 1$. Let $f=X^q-X\in \F_q[X]$. Then $A(f)=G(f)\cong (C_p)^m$.
\end{lemma}
\begin{proof}
Let $x$ be a root of $f$ and let $F\coloneqq\F_q(x)$. Then $F$ is a Galois extension of $\F_q(t)$, because $x+\alpha$ is a root of $f(X)-t$ for every $\alpha\in \F_q$, and therefore $|A(f)|=q$. Since $f(X)-t$ is absolutely irreducible, both $A(f)$ and $G(f)$ act transitively on the set of roots of $f(X)-t$ and therefore it must be that $G(f)=A(f)$. If $\sigma\in A(f)$ and $r=x+\alpha$ is a root of $f-t$ for some $\alpha\in \F_q$, then $\sigma(r)=r+\beta$ for some $\beta\in \F_q$, and therefore $\sigma^p(r)=r$. It follows that $\sigma^p$ is the identity, and therefore $A(f)=(C_p)^m$.
\end{proof}

\begin{lemma}\label{cyclic_case}
Let $\ell$ be a prime and $q$ a prime power with $\ell\nmid q$. Let $f\in \F_q(X)$ be a degree $\ell$ polynomial. Then $G(f)=A(f)\cong C_\ell$ if and only if $\ell\mid q-1$ and $f=(X-a)^\ell+b$ for some $a,b\in \F_q$.

\end{lemma}
\begin{proof}
Sufficiency is obvious.

Conversely, suppose that $G(f)=A(f)\cong C_\ell$. Let $x$ be a root of $f(X)-t$ and $F\coloneqq \F_q(x)$. Then the ramification in $F$ is always tame, and hence the Riemann-Hurwitz formula implies that there must be a finite place of $\overline{\F}_q(t)$ that ramifies in $F$. Let this place correspond to $b\in\overline{\F}_q$. Then $f(X)-b$ must factor as $(X-a)^\ell$ for some $a\in \overline{\F}_q$. Comparing the coefficients of the linear terms, it follows immediately that $a\in \F_q$, and hence $b\in \F_q$. But then $f=(X-a)^\ell+b$, and in order to have $G(f)=A(f)$ the field of constants $F\cap\overline{\F}_q$ must be $\F_q$, implying immediately that $\ell\mid q-1$ because certainly $F$ contains a primitive $\ell$-th root of unity.
\end{proof}

\begin{theorem}

The following holds:
\begin{enumerate}
\item $G_2(q)=2$ for every $q$.
\item $G_3(q)=\begin{cases}3 & \mbox{if }3\mid q(q-1)\\
			6 & \mbox{otherwise} \end{cases}$.
\end{enumerate}

\end{theorem}
\begin{proof}
When $n=2$ and $q$ is odd, every quadratic $f\in \F_q[X]$ has $G(f)=A(f)\cong C_2$. When $q$ is even, by Lemma \ref{additive_case} if $f=X^2+X\in \F_q(X)$ we have $G(f)=A(f)=2$.

When $n=3$ and $3\mid q$, by Lemma \ref{additive_case} for $f=X^3-X$ we have $G(f)=A(f)\cong C_3$. When $3\mid q-1$ for $f=X^3$ we have $G(f)=A(f)\cong C_3$. When $3\nmid q(q-1)$ by Lemma \ref{cyclic_case} we cannot have $G(f)=A(f)=C_3$. The only other transitive group inside $S_3$ is $S_3$ itself, and hence $G_3(q)=6$.
\end{proof}

\begin{theorem}
$$G_4(q)=\begin{cases} 24 & \mbox{if }q=2\\
4 & \mbox{if $4\mid q-1$ or $q=2^m$ for some $m>1$}\\
				8 & \mbox{otherwise}\\\end{cases}.$$
\end{theorem}
\begin{proof}
Recall that the transitive groups of degree $4$ are: $C_4,C_2\times C_2,D_4,A_4$ and $S_4$. Here the non-trivial elements of $C_2\times C_2$ are products of two disjoint transpositions, and therefore this copy of $C_2\times C_2$ is contained in $A_4$.

If $q$ is even and greater than $2$, then there exist distinct elements $\alpha_1,\alpha_2,\alpha_3\in \F_q^*$ such that $\alpha_1+\alpha_2+\alpha_3=0$. Now let $f=X(X+\alpha_1)(X+\alpha_2)(X+\alpha_3)$; this is of the form $X^4+aX^2+bX$ for some $a,b$ with $b\neq 0$. Therefore if $x$ is a root of $f(X)-t$, then all other roots are of the form $x+\alpha_i$ for some $i\in\{1,2,3\}$. It follows that $\F_q(x)/\F_q(t)$ is a Galois extension of degree $4$, and therefore $G(f)=A(f)$ and $G_2(q)=4$. In fact it is easy to see that these monodromy groups are isomorphic to $C_2\times C_2$: if $\sigma$ is any element of $G(f)$ and $r$ is any root of $f-t$, then $\sigma(r)=r+\alpha$ for some $\alpha\in \F_q$, and hence $\sigma^2(r)=r$, showing that non-trivial elements have order $2$.

If $q=2$, a quick search shows that the only two separable polynomials $f$ with $A(f)\neq S_4$ are $X^4+X$ and $X^4+X^2+X$. However, the first one has $A(f)\cong D_8$ and $G(f)\cong C_2\times C_2$, while the second one has $A(f)\cong A_4$ and $G(f)\cong C_2\times C_2$. Hence $G_4(2)=24$.

If $4\mid q-1$, then for $f=X^4$ we have $G(f)=A(f)\cong C_4$, and therefore $G_4(q)=4$.

Finally, suppose $q$ is odd and $4\nmid q-1$. If $A(f)\leq A_4$, then the discriminant of $f-t$ is a square in $\F_q(t)$. However, such discriminant is always a polynomial of degree $3$ in $t$, so $A(f)\leq A_4$ can never happen (and therefore, in particular, we cannot have $A(f)\cong C_2\times C_2$). If it was $A(f)=G(f)\cong C_4$, then since the ramification in the splitting field $F$ of $f-t$ is tame, by Riemann-Hurwitz there is a finite place of $\overline{\F}_q(t)$ that ramifies in $F$. Let $R$ be a place of $F$ lying over it. Then the decomposition group $D(R|P)$ is either $C_4$ or $C_2$. In the former case, for some $t_0\in \overline{\F}_q$ the polynomial $f-t_0$ factors as $(X-a)^4$, and this leads to a contradiction as in the proof of Lemma \ref{cyclic_case}. In the latter case, up to translations we have $f-t_0=X^2(X-a)^2$ for some $a,t_0\in \overline{\F}_q$ with $a\neq 0$, because the element of order $2$ in $C_4$ is a product of two disjoint transpositions. This implies that $f$ equals the composition $g\circ h$, where $g=X^2$ and $h=X(X-a)$, and it is a well-known fact (see for example \cite{ferra}) that for the Galois group of $g\circ h-t$ to be smaller than $D_4$ one would need $t$ and $a^2/4-t$ to be linearly dependent in the $\F_2$-vector space $\F_q(t)^*/{\F_q(t)^*}^2$, which clearly does not hold since $a\neq 0$. Hence $G_4(q)\geq 8$. On the other hand, again for the same well-known reasons one has that if $f=X^4+bX^2$ for some $b\neq 0$ then $G(f)=A(f)=D_4$. Hence $G_4(q)=8$.
\end{proof}

\section{Degree $5$: $\AGL_1(\mathbb{F}_5)$ never occurs}
In this section we will prove that if $5\nmid q$ then there exists no degree $5$ polynomial $f$ with $G(f)\cong \AGL_1(\F_5)$. From now on, we let $M$ be the splitting field of $f-t$ over $\F_q(t)$. If $P$ is a place of $\overline{\F}_q(t)$ and $R$ is a place of $M$ lying above it, we denote by $D(R|P)$ the corresponding decomposition group. For every $t_0\in \overline{\F}_q$, we denote by $P_{t_0}$ the corresponding place of $\overline{\F}_q(t)$.

We start with a preliminary lemma.
\begin{lemma}\label{factorization_lemma}
Let $q$ be a prime power with $5\nmid q$. Let $f\in \F_q[X]$ be a degree 5 polynomial and assume that $G(f)\cong \AGL_1(\F_5)$. Then there are $a,b,t_0\in \overline{\F}_q$ with $a\neq b$ such that $f-t_0=(X-a)^4(X-b)$.
\end{lemma}

\begin{proof}
We start by showing that $M = \overline{\mathbb{F}}_q(x,x') = \overline{\mathbb{F}}_q(x) \overline{\mathbb{F}}_q(x')$ for any two roots $x \neq x'$ of $f(X) - t$ in $M$. Observe that $\overline{\mathbb{F}}_q(t)(x) = \overline{\mathbb{F}}_q(x)$ since $t = f(x) \in \overline{\mathbb{F}}_q(x)$ (and similarly for $x'$) and write $F = \overline{\mathbb{F}}_q(x)$ and $F' = \overline{\mathbb{F}}_q(x')$. Clearly we have $[F : \overline{\mathbb{F}}_q(t)] = [F' : \overline{\mathbb{F}}_q(t)] = \deg f = 5$. Because $G(f) \cong \AGL_1(\mathbb{F}_5)$ is $2$-transitive, the stabilizer $G_x \subseteq G$ of $x$ acts transitively on the four other roots of $f - t$. In particular, since $G_x$ is the Galois group of $(f(X) - t)/(X-x)$ over $\overline{\mathbb{F}}_q(t)$, and since the orbit of $x'$ under the action of $G_x$ is a set of size $4$, we have $[FF' : F] = [G : G_x] = 4$. By definition $M \supseteq FF'$, so since $|G| = 20$ we have $M = FF' = \overline{\mathbb{F}}_q(x,x')$.

Now, let $P_{\infty}$ be the place at infinity of $\overline{\mathbb{F}}_q(t)$ and let $R_{\infty}$ be a place of $M$ lying over $P_{\infty}$. Let $Q_{\infty} = R_{\infty} \cap F$ and consider the ramification index $e(Q_{\infty}|P_{\infty})$ of $Q_{\infty}$ over $P_{\infty}$. Recalling that $f(x) = t$ in $F$, we have
$$
\upsilon_{Q_{\infty}}(f(x)) = \upsilon_{Q_{\infty}}(t)  = \upsilon_{P_{\infty}}(t) \cdot e(Q_{\infty}|P_{\infty})  = - e(Q_{\infty}|P_{\infty}),
$$
since $P_{\infty}$ is a pole of order $1$ of $t$. On the other hand, we also have $\upsilon_{Q_{\infty}}(f(x)) = \deg f \cdot \upsilon_{Q_{\infty}} (x) = -5$ by the strict triangle inequality. This yields $e(Q_{\infty}|P_{\infty}) = 5$, and an identical argument applied to $Q'_{\infty} = R_{\infty} \cap F'$ yields $e(Q'_{\infty}|P_{\infty}) = 5$. Since we have seen that $M$ is the compositum of the fields $F$ and $F'$ (both of which are tame extensions of $\overline{\mathbb{F}}_q(t)$ as we are working in characteristic $\neq 5$), it now follows from Abhyankar's Lemma \cite[Theorem 3.9.1]{stich} that $e(R_{\infty}|P_{\infty}) = \textup{lcm}\{e(Q_{\infty}|P_{\infty}), e(Q'_{\infty}|P_{\infty})\} = 5$. Thus the decomposition group $D(R_{\infty}|P_{\infty})$ is a group of order $5$, and hence it is isomorphic to $C_5$.

Next, we claim that there must be some $t_0 \in \overline{\mathbb{F}}_q$ such that for any place $R$ of $M$ lying over $P_{t_0}$, the decomposition group $D(R|P_{t_0})$ is isomorphic to $C_4$. To prove it, start by noticing that there must be some $t_0 \in \overline{\mathbb{F}}_q$ such that the decomposition group of any place of $M$ lying above it contains a cycle of order $4$. In fact, consider the subset of $G(f)$ of elements of even order that belong to some decomposition group: this contains no transpositions because the transitive copy of $\AGL_1(\F_5)$ inside $S_5$ contains no transpositions, and on the other hand if all such elements had order $2$ then they would all be products of two transpositions. However the decomposition groups generate $G(f)$,\footnote{This is because if $G$ is the subgroup of $G(f)$ generated by all the decomposition subgroups, then $M^G$ is an unramified extension of $\overline{\F}_q(t)$, and there are no non-trivial such extensions.} and in this latter case it would follow that $G(f)\leq A_5$, which is false once again. So let $t_0 \in \overline{\mathbb{F}}_q$ be such that for some place $R$ of $M$ lying above $P_{t_0}$, the decomposition group contains a cycle of order $4$. If we had $C_4\subset D(R|P_{t_0})\neq C_4$, it would follow $D(R|P_{t_0})\cong G(f)$ by the maximality of $C_4$ in $G(f)$. But then $R|P_{t_0}$ would be totally ramified, and hence $f(X) - t_0 = (X-a_0)^5$ for some $a_0 \in \overline{\mathbb{F}}_q$. Since the field of constants of $M/\overline{\mathbb{F}}_q(t)$ is trivial, this factorization implies $5 \mid q-1$. But then $G \cong C_5$, an immediate contradiction. Thus $D(R|P_{t_0}) \cong C_4$.

To conclude the proof, notice that specializing at the place $P_{t_0}$ and applying the Dedekind-Kummer Theorem \cite[Theorem 3.3.7]{stich} allows us to write $f(X) - t_0 = (X-a)^4(X-b)$ for $a,b \in \overline{\mathbb{F}}_q$ with $a \not = b$.
\end{proof}
We are now ready to prove that $\AGL_1(\F_5)$ cannot occur as a geometric monodromy group. The proof will require separate arguments for even and odd characteristics.
\begin{theorem}\label{AGL}
Let $q$ be a prime power with $5\nmid q$ and $f\in \F_q[X]$ a polynomial of degree 5. Then $G(f)\not\cong\AGL_1(\F_5)$.
\end{theorem}

\begin{proof}
Assume by contradiction that $G(f)\cong\AGL_1(\F_5)$. By Lemma \ref{factorization_lemma}, there are $t_0,a,b\in \overline{\F}_q$ such that $f-t_0=(X-a)^4(X-b)$. Since $\gal(f(X)-t \mid \overline{\F}_q(t))\cong\gal(f(X-c)-(t-d) \mid \overline{\F}_q(t-d))$ for every $c,d\in\overline{\F}_q$, we can assume without loss of generality that $t_0=a=0$ and $b\neq 0$.

\vspace{3mm}

\textbf{Odd characteristic.} Computing $f'(X) = X^3(5X-4b)$ shows that $f'(4b/5) = 0$, so for $t_1 = f(4b/5) \in \overline{\mathbb{F}}_q$, we see that $f(X) - t_1$ is divisible by $(X-4b/5)^2$. Furthermore, since $X = 4b/5$ is not a root of $f''(X) = 4X^2(5X-3b)$, it follows that $X=4b/5$ is precisely a double root of $f(X) - t_1$. Notice that $f(4b/5)\neq 0 = f(4b/5) - t_1$ and hence $t_1\neq 0$, so $X=0$ is not a root of $f(X) - t_1$. This implies that the only repeated root of $f(X) - t_1$ is $X = 4b/5$ since the only roots of $f'(X)$ are $X = 0$ and $X=4b/5$. In other words, we can write $f(X) - t_1 = (X - 4b/5)^2 (X - x_1) (X - x_2) (X - x_3)$ for pairwise distinct elements $4b/5, x_1, x_2, x_3 \in \overline{\mathbb{F}}_q$. Finally, let $R$ be any place of $M$ lying over $P_{t_1}$. Then the previous factorization shows that there is a transposition in $D(R|P_{t_1})\leq G(f)$. But $G(f)$ contains no transpositions, and we have a contradiction.

\vspace{3mm}

\textbf{Even characteristic.} From now on, we let $x$ be a root of $f(X)-t$ in $M$ and $f\coloneqq \overline{\F}_q(x)$. Fix a place $R$ of $M$ lying above the place $P_0$ of $\overline{\F}_q(t)$. The natural action of $C_4 \cong D(R|P_0) \subseteq G$ on the set of roots of $f(X) - t$ yields orbits of sizes $4$ and $1$, so there must be two places $Q_0$ and $Q_1$ of $F$ lying over $P_0$ with ramification indices $e(Q_0|P_0) = 4$ and $e(Q_1|P_0) = 1$ by \cite[Lemma 2.1]{BM}. Let $R_0$ be a place of $M$ lying over $Q_0$. We have just seen that $4 = |D(R_0|P_0)| = e(R_0|P_0) = e(R_0|Q_0)\cdot e(Q_0|P_0)$, so it follows that $e(R_0|Q_0) = 1$.

Before proceeding, we introduce the following notation: given a function field $K$ and a place $P$ of $K$, we will write $\hat K_P$ to denote the completion of $K$ at $P$ with respect to the $P$-adic metric. In particular, $\hat F_{Q_0} = \overline{\mathbb{F}}_q((x))$ and $\hat{\overline{\mathbb{F}}}_q(t)_{P_0} = \overline{\mathbb{F}}_q((t))$. Using a well-known number theoretical fact (see for example \cite[Proposition II.9.6]{neukirch}), we have $\textup{Gal}(\hat M_{R_0} \mid \overline{\mathbb{F}}_q((t))) \cong D(R_0|P_0) \cong C_4$. Observe that $\hat M_{R_0} \supseteq \hat F_{Q_0} \supseteq \overline{\mathbb{F}}_q((t))$, so because $[\hat F_{Q_0} : \overline{\mathbb{F}}_q((t))] = e(Q_0|P) = 4$ and $[\hat M_{R_0} : \overline{\mathbb{F}}_q((t))] = e(R_0|P_0) = 4$ we have $\hat M_{R_0} = \hat F_{Q_0}$. In particular, $\hat F_{Q_0} / \overline{\mathbb{F}}_q((t))$ is a Galois extension. Denoting the local Galois group $\textup{Gal}(\hat F_{Q_0} \mid \overline{\mathbb{F}}_q((t)))$ by $\hat G$, we have $\hat G \cong C_4$.

The above shows that every root of $f(X) - t$ in $\hat M_{R_0}$ can be expressed as an element of $\hat F_{Q_0} = \overline{\mathbb{F}}_q((x))$, that is, as a Laurent series in $x$. We proceed by showing that if $z \not = x$ is any other root of $f(X) - t$, then we can write $z = x + ux^i$ for some $i \geq 2$ and $u \in \overline{\mathbb{F}}_q[[x]]$. First, recall that $f(X) - t = X^4(X-b) - t$ so that $b$ is a simple root of $f(X)$. Then by Hensel's lifting lemma there is some $\bar{b} \in \overline{\mathbb{F}}_q((t))$ such that we can write $f(X) - t = \bar f(X)(X-\bar b)$ over $\overline{\mathbb{F}}_q((t))$, where $\bar f(X) \in \overline{\mathbb{F}}_q((t))[X]$ and $\deg \bar f = 4$. Further, the polynomial $\bar f$ must be irreducible over $\overline{\mathbb{F}}_q((t))$ since otherwise we could write $\bar f(X) = \bar g(X) (X- r)$ for an irreducible $\bar g \in \overline{\mathbb{F}}_q((t))[X]$ and some $r \in \overline{\mathbb{F}}_q((t))$, or we could write $\bar f(X) = \bar h_1(X) \bar h_2(X)$ for two irreducible quadratic polynomials $\bar h_1, \bar h_2 \in \overline{\mathbb{F}}_q((t))[X]$ having distinct roots in $F_{Q_0}$ (since $f - t$ is separable). The former factorization implies $| \hat G | = 4 \mid 3!$, a clear contradiction, so assume the latter factorization holds and let $\hat H_1$ and $\hat H_2$ be the splitting fields of $\bar h_1(X)$ and $\bar h_2(X)$, respectively, in $F_{Q_0}$. Then $F_{Q_0} = H_1H_2$ so that $C_4 \cong \hat G = \textup{Gal}(H_1H_2 \mid \overline{\mathbb{F}}_q((t))) \subseteq \textup{Gal}(H_1 \mid \overline{\mathbb{F}}_q((t))) \times \textup{Gal}(H_2 \mid \overline{\mathbb{F}}_q((t))) \cong C_2 \times C_2$, another contradiction. Thus we conclude that $\hat G$ is the Galois group of $\bar f(X)$ over $\overline{\mathbb{F}}_q((t))$; in particular, we have that $\hat G$ acts transitively on the roots of $\bar f(X)$ in $F_{Q_0}$ and hence there is some automorphism $\tau \in \hat G$ of $M$ satisfying $\tau(x) = z$.

Observe that $\langle x \rangle$ is the unique maximal ideal of the ring $\overline{\mathbb{F}}_q[[x]]$, so since $\tau$ is an automorphism of $\overline{\mathbb{F}}_q((x))$ (and hence $\tau$ preserves maximal ideals) we must have $\langle \tau(x) \rangle = \tau \left( \langle x \rangle \right) = \langle x \rangle$. Then $z \equiv 0 \mod{\langle x \rangle}$ if and only if $\tau(x) \equiv 0 \mod{\langle \tau(x) \rangle}$, and the latter clearly holds. This allows us to write $z = cx + ux^i$ for some $c \in \overline{\mathbb{F}}_q$ and for some $i \geq 1$, and we can assume further that $i \geq 2$ since otherwise we could replace $cx$ by $c'x$ for an appropriate $c' \in \overline{\mathbb{F}}_q^*$ so that this holds. Now computing $\tau^4(x)$ by using $\tau(x) = cx + ux^i$ and comparing coefficients with $\tau^4(x) = x$ yields $c^4 = 1$ and hence $c = 1$ as we are working over a field with characteristic $2$. Thus we can write $z = x + ux^i$ for some $i \geq 2$ and some $u\in \overline{\F}_q[[x]]^*$.

Observe the following:
\begin{align*}
f(z) - t & = z^4 (z - b) - t \\
& = \left( x + ux^i \right)^4 (x + ux^i - b) - t \\
& = \left( x^4 + u^4x^{4i} \right) \left( x - b + ux^i \right) - t \\
& = x^4(x - b) - t + ux^{4+i} + u^4x^{4i}(x - b) + u^5x^{5i} \\
& = ux^{4+i} + u^4x^{4i}(x - b) + u^5x^{5i},
\end{align*}
where the last equality holds since $x$ is a root of $f(X) - t$. Let $A = ux^{4+i} + u^4x^{4i}(x - b) + u^5x^{5i} = f(z) - t$. Since $z$ was chosen to be another root of $f(X) - t$, we must have $A = 0$. But $\upsilon_{Q_0}(A) = \upsilon_{Q_0}(x) \cdot \min\limits_{i \geq 2} \{4+i, 4i, 5i\} = 4+i \not = \infty = \upsilon_{Q_0}(0)$, a contradiction since $u \not = 0$. Thus our initial assumption was false, so $G(f) \not \cong \textup{AGL}_1(\mathbb{F}_5)$.
\end{proof}

\section{Degree $5$: if $D_5$ or $A_5$ occurs, then $5\mid q^2-1$}

Assume $q$ is a prime power and $f\in \F_q[X]$ is a separable polynomial of degree $5$. Let $t$ be transcendental over $\F_q$ and let $M$ be the splitting field of $f(X)-t$. Let $A(f)$ and $G(f)$ be the arithmetic and geometric monodromy groups of $f$, respectively.

\begin{theorem}\label{A5}
Suppose that $5\nmid q$ and $A(f)\leq A_5$. Then $5\mid q^2-1$. In particular, if $A(f)\leq D_5$, then $q^2-1$.
\end{theorem}

\begin{proof}

The second assertion follows immediately from the fact that the transitive copy of $D_5$ inside $S_5$ lies inside $A_5$.

If $q$ is odd, just use the fact that $A(f)\leq A_5$ if and only if the discriminant of $f(X)-t$ is a square in $\F_q(t)$. When $f(X)$ is monic of degree 5, the discriminant has the form $5^5t^4+\sum_{i=0}^3a_it^i$. Hence $5$ needs to be a square in $\F_q$; this implies that either $q$ is an even power of a prime $p$ or, by quadratic reciprocity, that $q\equiv \pm 1\bmod 5$. In any case, $5\mid q^2-1$.

If $q=2^n$ for some $n\geq 1$, one needs to use the characteristic two analogue of the discrimininant, called \emph{Berlekamp discriminant} (see \cite{berlekamp}). If $k$ is a field of characteristic 2 and $g\in k[x]$ has degree $n$, the Berlekamp discriminant of $g$ is an element $\Delta\in k$, that can be effectively computed using the coefficients of $g$, that has the property that $\gal(g)\leq A_n$ if and only if the polynomial $X^2+X+\Delta$ has a root in $k$.

Now let $f=X^5+aX^4+bX^3+cX^2+dX\in \F_{2^n}[X]$. We will show that if $A(f)\leq A_5$ then $2\mid n$, and consequently $5\mid q^2-1$ once again. One can compute the Berlekamp discriminant $\Delta_f$ of $f-t$, seen as a polynomial over $\F_{2^n}(t)$, and see that this is given by an expression of the form $r(t)/s(t)^2$, where $r,s\in \F_{2^n}(t)$ are two monic polynomials with $\deg r=4$ and $\deg s=2$. Suppose that $A(f)\leq A_5$, and hence $X^2+X+\Delta_f$ has a root in $\F_{2^n}(t)$. Then there are coprime polynomials $u(t),v(t)\in \F_{2^n}[t]$, with $v(t)$ monic, such that $\Delta_f=(u(t)^2+u(t)v(t))/v(t)^2$. Therefore if $r,s$ share a common factor, this can only have degree $2$ or $4$, and if it has degree $2$ then it must be of the form $(t+\alpha)^2$ for some $\alpha$. Clearly if they share a factor of degree $4$ then $\Delta_f=r(t)/s(t)^2=1$, so that $X^2+X+1$ has a root in $\F_{2^n}(t)$ and consequently $2\mid n$. Otherwise looking at degrees one sees that it must be $\deg u=\deg v$ and $\deg(u^2+uv)=2\deg v$. If the leading coefficient of $u$ is $\delta\in \F_{2^n}$, these two conditions, together with the fact that $r(t)$ is monic of degree $4$, imply that $\delta^2+\delta+1=0$ and consequently that $2\mid n$.
\end{proof}

\begin{proof}[Proof of Theorem \ref{main_thm}]

First, suppose that $5\mid q(q-1)$. Then by Lemmas \ref{additive_case} and \ref{cyclic_case}, $G_5(q)=5$.

Now suppose that $5\nmid q(q-1)$. Then by Lemma \ref{cyclic_case} we have $G_5(q)>5$. If  $5\mid q+1$, it is known (see for example \cite[Section 3]{cohen}) that degree $5$ Dickson polynomials of the first kind, e.g.\ $f=X^5-5X^3+5X$, satisfy $G(f)=A(f)=D_5$. Hence $G_5(q)=10$.

Finally, suppose that $5\nmid q(q^2-1)$. Then by Lemma \ref{cyclic_case} and Theorems \ref{AGL} and \ref{A5} we cannot have $G_5(q)=5,10,20$ or $60$. Hence $G_5(q)=120$.
\end{proof}

\begin{remark}
Notice that the $q/10$ asymptotic when $5|q+1$ was in fact obtained in \cite{liu2020constructions} using Dickson Polynomials and an independent approach.
\end{remark}

\section{Computational examples}

Let us show with a couple of explicit examples how the number of totally split places compares to the theoretical estimate given by Proposition \ref{split_places}. We pick examples with $G_5(q)=120$; for each of these values of $q$ we pick polynomials $f$ with $G(f)=A(f)\cong S_5$ for $5\nmid q(q^2-1)$. As proved in Theorem \ref{main_thm}, it is not possible to do better for these $q$'s.

In order to construct polynomials with geometric monodromy (and therefore also arithmetic monodromy) $S_5$, one can use the following well-known group theoretical fact (see \cite{gallagher}): if $G\leq S_5$ is a transitive subgroup containing a transposition and a cycle of prime length $\ell>2$, then $G=S_5$. In order to force the geometric monodromy group to contain two such elements it is enough, by ramification arguments, to pick $g(X)\in \F_q[X]$ irreducible of degree 3, and set $f=X^2g(X)$.

\begin{table}[h!]
\subfloat[$f=X^2(X^3+X+1)$]{
\begin{tabular}{c|c|c}

	$\qquad q\qquad$ & $\mathcal T^1_{split}(f)$ & $\lfloor\frac{q}{120}\rfloor$ \\
	\hline
$2^{13} $ &  78&        68\\
$2^{15} $ &  278&       273\\
$2^{17} $ &  1088&      1092\\
$2^{19} $ &  4332&      4369\\
   \hline
\end{tabular}}\qquad\qquad
%\subcaption{$f=X^2(X^2+X+1)$} \label{tab:f1}
\subfloat[$f=X^2(X^3-X+1)$]{
\begin{tabular}{c|c|c}

	$\qquad q\qquad$ & $\mathcal T^1_{split}(f)$ & $\lfloor\frac{q}{120}\rfloor$ \\
	\hline
$3^{7}$ &  21&       18\\
$3^{9}$ &  159&     164\\
$3^{11}$ &  1474&      1476\\
$3^{13}$ &  13338&      13286\\
   \hline
\end{tabular}}

\subfloat[singlelinecheck=true][$f=X^2(X^3+X+3)$]{
\begin{tabular}{c|c|c}

	$\qquad q \qquad$ & $\mathcal T^1_{split}(f)$ & $\lfloor\frac{q}{120}\rfloor$ \\
	\hline
$19583$ & 156 & 163\\
19597 & 163 & 163\\
19687 & 155 & 164\\
19753 & 194 & 164\\
19793 & 179  & 164\\
19913 &  189 & 165\\
19927 & 160 & 166\\
19963 & 162 & 166\\
19993 & 156 & 166\\
19997 & 161 & 166\\
   \hline
\end{tabular}}
%\subcaption{$f=X^2(X^2+X+1)$} \label{tab:f1}
\end{table}

\bibliographystyle{plain}
\bibliography{bibliography}

\end{document}